\documentclass[]{eptcs}
\pdfoutput=1

\usepackage{amsmath} 
\usepackage{amsthm} 
\usepackage{amssymb}	
\usepackage{bold-extra} 
\usepackage{hyperref}
\usepackage{stmaryrd} 
\usepackage{gensymb} 
\usepackage{enumitem} 

\newtheorem{thm}{Theorem}
\newtheorem*{thm*}{Theorem}

\newtheorem*{prop*}{Proposition}
\newtheorem{lem}[thm]{Lemma}

\theoremstyle{definition}

\newcommand{\avg}[1]{\left\langle #1 \right\rangle}
\newcommand{\CC}{\mathbb{C}}

\newcommand{\RR}{\mathbb{R}}
\newcommand{\NN}{\mathbb{N}}

\newcommand{\ket}[1]{\left| #1 \right>} 
\newcommand{\bra}[1]{\left< #1 \right|} 
\newcommand{\braket}[2]{\left< #1 \middle| #2 \right>} 
 
\renewcommand{\t}[1]{\ensuremath{^{\otimes #1}}}

\newcommand{\GHZ}{\mathrm{GHZ}}

\newcommand{\cA}{\mathcal{A}}

\newcommand{\cE}{\mathcal{E}}
\newcommand{\cF}{\mathcal{F}}
\newcommand{\cG}{\mathcal{G}}
\newcommand{\cL}{\mathcal{L}}
\newcommand{\cM}{\mathcal{M}}
\newcommand{\cS}{\mathcal{S}}
\newcommand{\cT}{\mathcal{T}}
\newcommand{\cU}{\mathcal{U}}

\DeclareMathOperator{\Holant}{Holant}
\newcommand{\Hol}{\textsc{Holant}}
\newcommand{\Holp}[2][]{\textsc{Holant}^{ #1 }\left( #2 \right)}
\newcommand{\csp}{\#\textsc{CSP}}

\newcommand{\FP}{\textsf{FP}}
\newcommand{\sP}{\#\textsf{P}}

\newcommand{\etal}[0]{\emph{et al.}}

\title{A complete dichotomy for complex-valued Holant\textsuperscript{c}}
\author{Miriam Backens
\institute{School of Mathematics, University of Bristol, UK}
\email{m.backens@bristol.ac.uk}
}
\def\titlerunning{A complete dichotomy for complex-valued Holant\textsuperscript{c}}
\def\authorrunning{Miriam Backens}

\begin{document}

\maketitle

\begin{abstract}
 Holant problems are a family of counting problems on graphs, parametrised by sets of complex-valued functions of Boolean inputs.
 \Hol$^c$ denotes a subfamily of those problems, where any function set considered must contain the two unary functions pinning inputs to values 0 or 1.
 The complexity classification of Holant problems usually takes the form of dichotomy theorems, showing that for any set of functions in the family, the problem is either \sP-hard or it can be solved in polynomial time.
 Previous such results include a dichotomy for real-valued \Hol$^c$ and one for \Hol$^c$ with complex symmetric functions.
 
 Here, we derive a dichotomy theorem for \Hol$^c$ with complex-valued, not necessarily symmetric functions.
 The tractable cases are the complex-valued generalisations of the tractable cases of the real-valued \Hol$^c$ dichotomy.
 The proof uses results from quantum information theory, particularly about entanglement.
\end{abstract}

\section{Introduction}

Holant problems are a framework for the analysis of counting complexity problems defined on graphs.
They encompass and generalise other counting complexity frameworks like counting constraint satisfaction problems (\csp) \cite{cai_holant_2012,cai_complexity_2014} and counting graph homomorphisms \cite{cai_holant_2009,cai_graph_2010}.

A Holant instance is defined by assigning a function from a specified set to each vertex of a graph, with the edges incident on that vertex corresponding to inputs of the function.
The counting problem is a sum-of-products computation: multiplying all the function values and then summing over the different assignments of input values to the edges \cite{cai_holant_2009}.
Here, we consider complex-valued functions of Boolean inputs.
Throughout, all numbers are assumed to be algebraic.

Problems that can be expressed in the Holant framework include counting matchings or counting perfect matchings, counting vertex covers \cite{cai_holant_2009}, or counting Eulerian orientations \cite{huang_dichotomy_2016}.
A Holant problem can also be thought of as the problem of contracting a tensor network; from that perspective, each function corresponds to a tensor with one index for each input.

The main goal in the analysis of Holant problems is the derivation of dichotomy theorems, showing that all problems in a certain family are either polynomial time solvable or \sP-hard.
Families of Holant problems are often defined by assuming that the function sets contain specific functions, which are said to be `freely available'.
As an example, the problem $\csp(\cF)$ for a function set $\cF$ effectively corresponds to the Holant problem $\Holp{\cF\cup\{=_n\mid n\in\NN^*\}}$, where $(=_1):\{0,1\}\to\CC$ is the function that is 1 on both inputs, and, for $n\geq 2$, $(=_n):\{0,1\}^n\to\CC$ is the function satisfying:
\[
 (=_n)(x_1,x_2,\ldots,x_n) = \begin{cases} 1 &\text{if } x_1=x_2=\ldots=x_n \\ 0 & \text{otherwise.} \end{cases}
\]
The problem $\Holp[c]{\cF}$ is the Holant problem where the unary functions pinning edges to values 0 or 1 are available in addition to the elements of $\cF$:
\[
 \Holp[c]{\cF} = \Holp{\cF\cup\{\delta_0,\delta_1\}},
\]
with $\delta_0(0)=1$, $\delta_0(1)=0$ and conversely for $\delta_1$ \cite{cai_holant_2009}.
Another important family is \Hol$^*$, in which all unary functions are freely available \cite{cai_dichotomy_2011}.

Known Holant dichotomies include a full dichotomy for \Hol$^*$ \cite{cai_dichotomy_2011}, a dichotomy for \Hol$^c$ with symmetric functions, i.e.\ where all functions in the sets considered depend only on the Hamming weight of the input \cite{cai_holant_2012}, and a dichotomy for real-valued \Hol$^c$, where functions need not be symmetrical but must take values in $\RR$ instead of $\CC$ \cite{cai_dichotomy_2017}.
Both existing results about \Hol$^c$ are proved via dichotomies for \csp{} problems with complex-valued, not necessarily symmetric functions: in the first case, a dichotomy for general \csp{} problems, and in the second case, a dichotomy for \csp$_2^c$, a subfamily of \csp{} in which each variable must appear an even number of times and variables can be pinned to 0 or 1.

While many dichotomies have been derived for functions taking values in some smaller set, we consider complex-valued functions to be the natural setting for Holant problems.
This is motivated in part by connecting Holant problems to quantum computation, where complex numbers naturally arise: the problem of classically simulating a quantum circuit with fixed input and output states can immediately be expressed as a Holant problem.
The second justification for considering complex numbers is that many tractable sets find a more natural expression over $\CC$.
An example of this are the `affine functions' (see Section \ref{s:preliminary}): they were originally discovered as several distinct tractable sets for a smaller codomain, but their definition is vastly more straightforward when expressed in terms of complex values \cite{cai_complexity_2014}.

We therefore build on the existing work to derive a \Hol$^c$ dichotomy for complex-valued, not necessarily symmetric functions.
In the process, we employ notation and results from quantum information theory.
This approach was first used in a recent paper \cite{backens_new_2017} to derive a dichotomy for \Hol$^+$, in which four unary functions are freely available, including the ones available in \Hol$^c$:
\[
 \Holp[+]{\cF} = \Holp{\cF\cup\{\delta_0,\delta_1,\delta_+, \delta_-\}},
\]
where $\delta_+(x)=1$ for both inputs (i.e.\ it is the same as the unary equality function) and $\delta_-(x)=(-1)^x$.

A core part of quantum theory, and also of the quantum approach to Holant problems, is the notion of entanglement.
A pure\footnote{There are also \emph{mixed} quantum states, which have a different mathematical representation, and which are not considered here.} quantum state of $n$ qubits, the quantum equivalents of bits, is represented by a vector in the space $(\CC^2)\t{n}$, which consists of $n$ tensor copies of $\CC^2$.
Such a vector is called \emph{entangled} if it cannot be written as a tensor product of vectors from each copy of $\CC^2$.

An $n$-ary function $f:\{0,1\}^n\to\CC$ can be considered as a vector in $\CC^{2^n}$ by treating each input as an element of an orthonormal basis for that space and using the function values as coefficients in a linear combination of those basis vectors (cf.\ Section \ref{s:Holant_problems}).
This vector space $\CC^{2^n}$ is isomorphic to $(\CC^2)\t{n}$, allowing functions to be brought into correspondence with quantum states.
We thus call a function entangled if the associated vector is entangled.
Identifying this property in Holant problems lets us apply some of the large body of work on quantum entanglement to Holant problems \cite{popescu_generic_1992,gachechiladze_completing_2017,dur_three_2000,li_simple_2006}.
The resulting complexity classification of Holant problems remains non-quantum, we simply employ a different set of mathematical tools to analyse them.

In the \Hol$^+$ dichotomy, it was shown how to construct a gadget for a ternary entangled function, given an $n$-ary entangled function with $n\geq 3$ and using the freely-available unary functions.
Furthermore, in most cases it was shown to be possible to realise a ternary symmetric function from this \cite{backens_new_2017}.
We show how to adapt those constructions to the \Hol$^c$ framework, where only two unary functions are freely available.
This does not always work, yet all cases in which the construction fails turn out to be ones that were identified in the dichotomy for real-valued \Hol$^c$: in those cases, either the problem is tractable by the \Hol$^*$ dichotomy or it can be reduced to \csp$_2^c$.
With these adaptations, we extend the dichotomy theorem for real-valued \Hol$^c$ to arbitrary complex-valued functions.

In the following, Section \ref{s:Holant_problems} contains the formal definition of Holant problems and an overview over common strategies used in classifying their complexity.
We recap existing results in Section \ref{s:existing}.
The new dichotomy and its constituent lemmas are proved in Section \ref{s:dichotomy}.
Section \ref{s:conclusions} contains the conclusions and outlook.

\section{Holant problems}
\label{s:Holant_problems}

Holant problems are a framework for counting complexity problems defined on graphs, first introduced in \cite{cai_holant_2009}.
Let $G=(V,E)$ be a graph with vertices $V$ and edges $E$ and let $\cF$ be a set of complex-valued functions of Boolean inputs.
Throughout, when we refer to complex numbers we mean algebraic complex numbers.
Let $\pi$ be a function that assigns to each degree-$n$ vertex $v$ in the graph an $n$-ary function $f_v\in\cF$ and also assigns one edge incident on the vertex to each input of the function.
This determines a complex value associated with the tuple $(\cF, G, \pi)$, called the \emph{Holant} and defined as follows:
\[
 \Holant_{(\cF, G, \pi)} = \sum_{\sigma:E\to\{0,1\}} \prod_{v\in V} f\left( \sigma |_{E(v)} \right).
\]
Here, $\sigma$ is an assignment of a Boolean value to each edge in the graph and $\sigma|_{E(v)}$ is the restriction of $\sigma$ to the edges incident on vertex $v$.
The tuple $(\cF, G, \pi)$ is called a \emph{signature grid}.

The associated counting problem is $\Holp{\cF}$: given a signature grid $(\cF, G, \pi)$ for the fixed set of functions $\cF$, find $\Holant_{(\cF, G, \pi)}$.

It is often useful to think of the functions, also called signatures, as vectors or tensors \cite{cai_valiants_2006}.
The $n$-ary functions can be put in one-to-one correspondence with vectors in $\CC^{2^n}$ as follows: pick an orthonormal basis for $\CC^{2^n}$ and label its elements $\{ \ket{x} \}_{x\in\{0,1\}^n}$, i.e. each basis vector is labelled by one of the $2^n$ $n$-bit strings.\footnote{The $\ket{\cdot}$ notation for vectors is the Dirac or bra-ket notation commonly used in quantum theory.}
Then assign to each $f:\{0,1\}^n\to\CC$ the vector $\ket{f} = \sum_{x\in\{0,1\}^n} f(x)\ket{x}$.
Conversely, any vector $\ket{\psi}\in\CC^{2^n}$ corresponds to an $n$-ary function $\psi:\{0,1\}^n\to\CC:: x\mapsto \braket{x}{\psi}$, where $\braket{\cdot}{\cdot}$ denotes the inner product of two vectors.\footnote{Strictly speaking, this notation refers to the complex inner product, i.e.\ $\bra{x}$ is the conjugate transpose of $\ket{x}$, but the distinction is irrelevant if all coefficients of $\ket{x}$ are real.}
The product of two functions of disjoint sets of variables corresponds to the tensor product of the associated vectors.
Where no confusion arises, we drop the tensor product symbol and sometimes even combine labels into a single `ket' $\ket{\cdot}$: for example, instead of writing $\ket{0}\otimes\ket{0}$, we may write $\ket{0}\ket{0}$ or $\ket{00}$.
If $g$ is a unary signature, we sometimes write $\bra{g}_l\ket{f}$ to indicate that the $l$-th input of $f$ is connected to a vertex with signature $g$.

The vector perspective is particularly useful for bipartite Holant problems.
Those arise on bipartite graphs if we choose to assign functions from two different signature sets to the vertices in the two different partitions.
Then the Holant becomes the inner product between two vectors corresponding to the two partitions.
Formally: let $G=(V,W,E)$ be a bipartite graph with vertex partitions $V$ and $W$, and let $\cF, \cG$ be two sets of signatures.
Suppose $\pi$ is a function that assigns elements of $\cF$ to vertices from $V$ and elements of $\cG$ to vertices from $W$ and otherwise acts as described above.
Then:
\[
 \Holant_{(\cF\mid\cG, G, \pi)} = \left( \bigotimes_{v\in V} (\ket{f_v})^T \right) \left( \bigotimes_{w\in W} \ket{g_w} \right), 
\]
where we assume the two tensor products are arranged so that the appropriate components of the two vectors meet.
A bipartite Holant problem over signature sets $\cF$ and $\cG$ is denoted $\Holp{\cF\mid\cG}$.

Any Holant instance can be made bipartite without changing the value of the Holant by inserting an additional vertex in the middle of each edge and assigning it the binary equality signature $=_2$.
Thus:
\[
 \Holp{\cF} \equiv_T \Holp{\cF\mid\{=_2\}},
\]
i.e.\ the two problems have the same complexity.
Formally, we write $A\leq_T B$ if there exists a polynomial-time reduction from problem $B$ to problem $A$.
We write $A\equiv_T B$ if $(A\leq_T B) \wedge (B\leq_T A)$.

In the following, we use the function and vector perspectives on signatures interchangeably.

\subsection{Complexity classification}

Most complexity results about the Holant problem take the form of \emph{dichotomies}, showing that for all signature sets in a specific family, the problem is either \sP-hard or in \FP.
Such a dichotomy is not expected to be true for all counting complexity problems, i.e.\ it is assumed that there are problems in $\sP\setminus\FP$ which are not \sP-hard \cite{cai_dichotomy_2011}.

A number of polynomial-time reduction techniques are commonly used in Holant problems.

The technique of \emph{holographic reductions} is the origin of the name Holant.
Let $M$ be a 2 by 2 invertible complex matrix and define $M\circ f = M\t{\operatorname{arity}(f)}\ket{f}$, where $M\t{1}=M$ and $M\t{n+1}=M\otimes M\t{n}$.
Furthermore, let $M\circ\cF = \{M\circ f \mid f\in\cF \}$; this is called a \emph{holographic transformation}.
Consider a bipartite graph and signature sets $\cF$ and $\cG$.
These satisfy:
\[
 \Holp{\cF\mid\cG} \equiv_T \Holp{M\circ\cF \mid (M^{-1})^T\circ\cG }.
\]
This is because the two transformations cancel out in the inner product: let $(\cF\mid\cG, G, \pi)$ be the set-up for a Holant instance and let $(M\circ\cF\mid (M^{-1})^T\circ\cG, G, \pi')$ be its analogue where $\pi'$ assigns $M\circ f$ whenever $\pi$ assigns $f$ on vertices of the left partition, and $(M^{-1})^T\circ g$ instead of $g$ on the right.
Then in fact $\Holant_{(\cF\mid\cG, G, \pi)} = \Holant_{(M\circ\cF\mid (M^{-1})^T\circ\cG, G, \pi')}$;
this is Valiant's Holant Theorem \cite{valiant_holographic_2008}.

A second technique is that of \emph{gadgets}.
Consider a subgraph of the original graph, which is connected to the larger graph by $n$ edges.
This subgraph can be replaced by a single vertex with an appropriate signature without changing the value of the overall Holant.
Conversely, consider some set of signatures $\cF$ and suppose $g$ is an $n$-ary signature with $g\notin\cF$.
If there exists some subgraph with signatures taken from $\cF$ such that the effective signature for that subgraph is $g$, then \cite{cai_dichotomy_2011}:
\[
 \Holp{\cF\cup\{g\}} \leq_T \Holp{\cF}.
\]
This is because any occurrence of a vertex with signature $g$ can be replaced by the subgraph without changing the value of the overall Holant.
The graph increases in size by an amount linear in the number of occurrences of $g$, as the subgraph contains a constant number of vertices, and the resulting signature grid contains only signatures from $\cF$.
Thus if $\Holp{\cF}$ can be solved in polynomial time, so can $\Holp{\cF\cup\{g\}}$.
We say $g$ is \emph{realisable} over $\cF$.
As multiplying a signature by a non-zero constant does not change the complexity of a Holant problem, we also consider $g$ realisable if we can construct a gadget with effective signature $cg$ for some $c\in\CC\setminus\{0\}$.

In bipartite signature grids, we may distinguish between left-hand side (LHS) gadgets and right-hand side (RHS) gadgets, which can be used as if they are signatures for the left and right partitions, respectively.

Finally, there is the technique of polynomial interpolation \cite{cai_holant_2009}.
This works for signatures of any arity, though for simplicity we will explain the process as applied to a unary signature.
Let $\cF$ be a set of signatures and let $\ket{g}=x\ket{0}+y\ket{1}$ be such that $g\notin\cF$ and there is no gadget over $\cF$ with effective signature $g$.
Consider a signature grid $\Omega = (\cF\cup\{g\}, G, \pi)$ in which $g$ is assigned $n$ times.
By pulling out the dependence on $g$, the corresponding Holant can be written as:
\[
 \Holant_{\Omega} = \sum_{j+k=n} c_{jk} x^j y^k,
\]
where the $c_{jk}$ are (as yet unknown) constants.
Now suppose it is possible to find polynomially many gadgets over $\cF$, having distinct unary signatures $\{a_m\ket{0}+b_m\ket{1}\}_{m\in\NN}$, and with the size of each gadget bounded by a polynomial.
By replacing all occurrences of $g$ in $\Omega$ with $a_m\ket{0}+b_m\ket{1}$, we construct a new signature grid $\Omega_m$ which only uses signatures from $\cF$.
Its Holant is:
\[
 \Holant_{\Omega_m} = \sum_{j+k=n} c_{jk} a_m^j b_m^k.
\]
If we can determine the Holant values for sufficiently many distinct $\Omega_m$ with sufficiently nice properties, we can set up a system of linear equations for the coefficients $c_{jk}$, and thus compute those in time polynomial in $n$.
Those coefficients, in turn, can be used to compute $\Holant_{\Omega}$, again in polynomial time.
In this case, we say that $g$ can be \emph{interpolated} over $\cF$ and conclude:
\[
 \Holp{\cF\cup\{g\}} \leq_T \Holp{\cF}.
\]
A more detailed definition can be found in \cite{cai_holant_2009}.

\subsection{Properties of signatures}
\label{s:properties}

We now introduce some definitions and terminology that will be used throughout the paper.

A signature is called \emph{symmetric} if its value as a function depends only on the Hamming weight of the inputs -- in other words, it is invariant under any permutation of the inputs.
Symmetric functions are often written in the short-hand notation $f=[f_0,f_1,\ldots, f_n]$, where $f_k$ is the value $f$ takes on inputs of Hamming weight $k$.

Using language from quantum theory, a signature is \emph{entangled} if it cannot be written as a tensor product of unary signatures.
For example, $\ket{01}+\ket{11}$ is not entangled because it can be rewritten as $(\ket{0}+\ket{1})\otimes\ket{1}$.
On the other hand, the binary equality signature $\ket{00}+\ket{11}$ is entangled.
If $k\geq 2$, a $k$-ary signature can be partially decomposable into a tensor product, e.g. $\ket{0}\otimes(\ket{00}+\ket{11})$.
We say a signature is \emph{genuinely entangled} if there is no way of decomposing it as a tensor product of signatures of any arity.
When the meaning is clear from context, we may sometimes drop the word `genuinely'.
A genuinely entangled signature of arity at least 3 is said to be \emph{multipartite entangled} (as opposed to the bipartite entanglement in a signature of arity 2).
A non-genuinely entangled signature has multipartite entanglement if it has a tensor factor corresponding to a genuinely entangled signature of arity at least 3.

Among genuinely entangled ternary signatures, we distinguish two types, also known as `entanglement classes' \cite{dur_three_2000}.
Let:
\[
 \ket{f} = a_0\ket{000} + a_1\ket{001} + a_2\ket{010} + a_3\ket{011} + a_4\ket{100} + a_5\ket{101} + a_6\ket{110} + a_7\ket{111},
\]
where $a_0,\ldots,a_7\in\CC$.
Then $\ket{f}$ has GHZ type if the following polynomial in the coefficients is non-zero:
\[
 (a_0a_7 - a_2a_5 + a_1a_6 - a_3a_4)^2 - 4(a_2a_4-a_0a_6)(a_3a_5-a_1a_7).
\]
The signature $\ket{f}$ has $W$ type if the above polynomial is zero, and furthermore the following expression is true:
\[
 ((a_0a_3\neq a_1a_2) \vee (a_5a_6\neq a_4a_7)) \;\wedge\; ((a_1a_4\neq a_0a_5) \vee (a_3a_6\neq a_2a_7)) \;\wedge\; ((a_3a_5\neq a_1a_7) \vee (a_2a_4 \neq a_0a_6)).
\]
If the polynomial is zero and the above expression evaluates to false, then the signature is not genuinely entangled \cite{li_simple_2006}.

All GHZ-type signatures can be transformed to $\ket{\GHZ}=\ket{000}+\ket{111}$ by some local holographic transformation, i.e. if $\ket{f}$ has GHZ type then there exist 2 by 2 invertible matrices $A,B,C$ such that $\ket{\GHZ}=(A\otimes B\otimes C)\ket{f}$.
Similarly all $W$-type signatures can be transformed to $\ket{W}=\ket{001}+\ket{010}+\ket{100}$ by a local holographic transformation \cite{dur_three_2000}.

We also consider generalised GHZ signatures of the form $\ket{\GHZ_n}=\ket{0}\t{n}+\ket{1}\t{n}$ for $n\in\NN^*$.
These are the same as the $n$-ary equality signatures, i.e.\ $\ket{\GHZ_n}=\ket{=_n}$.

There are many other types of genuinely entangled signatures for higher arities \cite{lamata_inductive_2006,verstraete_four_2002}, but those are not directly relevant to this paper.

Given a set of signatures that contains multipartite entanglement, we can assume without loss of generality that we have a genuinely multipartite-entangled signature.
To see this, consider a non-zero signature $\ket{\psi}$ that has multipartite entanglement, and suppose $\ket{\psi}=\ket{\psi_1}\otimes\ket{\psi_2}$.
Then at least one of the tensor factors must have multipartite entanglement, assume this is $\ket{\psi_1}$.
Now, $\ket{\psi}$ is non-zero, so $\braket{x}{\psi_2}$ must be non-zero for some bit string $x$.
Thus we can realise $\ket{\psi_1}$ by connecting all inputs associated with $\ket{\psi_2}$ to $\ket{0}$ or $\ket{1}$, as appropriate.

\section{Existing results}
\label{s:existing}

It is difficult to determine the complexity of the general Holant problem.
Thus, all existing dichotomies make use of one or more simplifying assumptions: either they assume the availability of certain signatures in all signature sets considered \cite{cai_dichotomy_2011,cai_valiants_2006}, or they only consider signature sets containing functions taken from more restricted families, e.g.\ symmetric functions \cite{cai_holant_2012,cai_complete_2013} or functions taking only real \cite{cai_dichotomy_2017} or even non-negative real values \cite{lin_complexity_2016}.

Among others, the following variants of the Holant problem have been considered:
\begin{itemize}
 \item $\Holp[*]{\cF} = \Holp{\cF\cup\cU}$, where $\cU$ is the set of all unary signatures \cite{cai_dichotomy_2011},
 \item $\Holp[+]{\cF} = \Holp{\cF\cup\{\ket{0},\ket{1},\ket{+},\ket{-}\}}$, where $\ket{\pm}=\ket{0}\pm\ket{1}$ \cite{backens_new_2017}, and
 \item $\Holp[c]{\cF} = \Holp{\cF\cup\{\ket{0},\ket{1}\}}$ \cite{cai_holant_2012,cai_dichotomy_2017}.
\end{itemize}
Several variants of complex-weighted Boolean counting constraint satisfaction problems (\csp) have also been expressed in the Holant framework.
These include:
\begin{itemize}
 \item $\csp(\cF) = \Holp{\cF \mid \cG}$, where $\cG$ is the set containing all equality signatures \cite{cai_complexity_2014}:
  \[
   \cG = \left\{ \ket{0}\t{n}+\ket{1}\t{n} \,\middle|\, n\in\NN^* \right\} = \{ \ket{0}+\ket{1}, \ket{00}+\ket{11}, \ket{000}+\ket{111}, \ldots \},
  \]
 \item $\#\text{R$_3$-CSP}(\cF) = \Holp{\cF\mid\{\ket{0}+\ket{1}, \ket{00}+\ket{11}, \ket{000}+\ket{111}\}}$ \cite{cai_complexity_2014},
 \item $\csp_2^c(\cF) = \Hol( \cF \mid \{ \ket{0},\ket{1}\}\cup\{ \ket{0}\t{2n}+\ket{1}\t{2n} \mid n\in\NN^* \} )$ \cite{cai_dichotomy_2017}.
\end{itemize}
Unlike the general \csp{} framework, $\#\text{R$_3$-CSP}$ problems contain only equality signatures of arity three or less.
The \csp$_2^c$ problems contain only equality signatures of even arity, but also the signatures pinning inputs to 0 or 1, respectively. 

Existing results include full dichotomies for \Hol$^*$ \cite{cai_dichotomy_2011}, \Hol$^+$ \cite{backens_new_2017}, and all three \csp{} variants mentioned above \cite{cai_complexity_2014,cai_dichotomy_2017}.
There are also dichotomies for \Hol$^c$ with symmetric complex-valued signatures \cite{cai_dichotomy_2011}, \Hol$^c$ with arbitrary real-valued signatures \cite{cai_dichotomy_2017}, \Hol{} with symmetric complex-valued signatures \cite{cai_complete_2013}, and \Hol{} with arbitrary non-negative real-valued signatures \cite{lin_complexity_2016}.

We build in particular on the dichotomies for \Hol$^+$ and real-valued \Hol$^c$, as well as on other results classifying the hardness of Holant problems with symmetric signatures on 3-regular graphs \cite{cai_holant_2012}.

\subsection{Preliminary definitions}
\label{s:preliminary}

The following definitions will be used throughout the dichotomy theorems.
Write:
\begin{equation}\label{eq:matrices}
 T=\begin{pmatrix}1&0\\0&e^{i\pi/4}\end{pmatrix}, \qquad X = \begin{pmatrix}0&1\\1&0\end{pmatrix} \qquad\text{and}\qquad K = \begin{pmatrix}1&1\\i&-i\end{pmatrix},
\end{equation}
where $i^2=-1$.
Then let:
\begin{itemize}
 \item $\cT$ be the set of all unary and binary signatures,
 \item $\cE$ the set of all signatures that are non-zero only on two inputs $x$ and $\bar{x}$, where $\bar{x}$ denotes the bit-wise complement of $x$,
 \item $\cM$ the set of all signatures that are non-zero only on inputs of Hamming weight at most 1,
 \item $\cA$ the set of all \emph{affine signatures}, i.e.\ functions of the form $f(x) = c i^{l(x)} (-1)^{q(x)} \chi$, where $c\in\CC$, $l(x)$ is a linear Boolean function, $q(x)$ is a quadratic Boolean function, and $\chi$ is the indicator function for an affine space, and
 \item $\cL$ the set of all signatures $f$ with the property that, for any bit string $x$ in the support of $f$:
  \[
   \left( \bigotimes_{j=1}^{\operatorname{arity}(f)} T^{x_j} \right) \ket{f} \in \cA.
  \]
\end{itemize}
Denote by $\avg{\cF}$ the closure of the signature set $\cF$ under tensor products.
It is straightforward to see that $\cA=\avg{\cA}$ and $\cL=\avg{\cL}$, i.e.\ these signature sets are already closed under tensor products.

If $n$ is a positive integer, we denote by $[n]$ the set $\{1,2,\ldots,n\}$.

\subsection{Dichotomies for Holant variants}

The Holant dichotomies build up on each other.
Dichotomies with fewer freely-available signatures refer to dichotomies for problems with more freely-available signatures, e.g.\ the dichotomy for \Hol$^+$ incorporates the dichotomy for \Hol$^*$.
We recap those theorems that will be used in deriving the dichotomy for \Hol$^c$.

\begin{thm}[Theorem 2.2, \cite{cai_dichotomy_2011}]\label{thm:Holant-star}
 Let $\cF$ be any set of complex valued functions in Boolean variables. The problem $\Holp[*]{\cF}$ is polynomial time computable if:
 \begin{itemize}
  \item $\cF\subseteq\avg{\cT}$, or
  \item $\cF\subseteq\avg{O\circ\mathcal{E}}$, where $O$ is a complex orthogonal 2 by 2 matrix, or
  \item $\cF\subseteq\avg{K\circ\mathcal{E}}$, or
  \item $\cF\subseteq\avg{K\circ\mathcal{M}}$ or $\cF\subseteq\avg{KX\circ\mathcal{M}}$.
 \end{itemize}
 In all other cases, $\Holp[*]{\cF}$ is \sP-hard.
\end{thm}

\begin{thm}[Theorem 6, \cite{cai_holant_2012}]\label{thm:symmetric_Holant-c}
 Let $\cF$ be a set of complex symmetric signatures. $\Holp[c]{\cF}$ is \sP-hard unless $\cF$ satisfies one of the following conditions, in which case it is tractable:
 \begin{itemize}
  \item $\Holp[*]{\cF}$ is tractable, or
  \item there exists a 2 by 2 matrix $S\in\cS$ such that $\cF\subseteq S\circ\cA$, where:
   \begin{equation}\label{eq:cS_definition}
    \cS = \left\{ S \,\middle|\, (S^T)\t{2}(\ket{00}+\ket{11}), S^T\ket{0}, S^T\ket{1} \in \cA \right\}.
   \end{equation}
 \end{itemize}
\end{thm}

A list of elements of $\cS$ is contained in an appendix of the full version of \cite{cai_holant_2012}.

\begin{thm}[Theorem 13, \cite{backens_new_2017}]
 Let $\cF$ be a set of complex-valued signatures. $\Holp[+]{\cF}$ is in \FP{} if $\cF$ satisfies one of the following conditions:
 \begin{itemize}
  \item $\Holp[*]{\cF}$ is in \FP, or
  \item $\cF\subseteq\cA$.
 \end{itemize}
 In all other cases, the problem is \sP-hard.
\end{thm}

\begin{thm}[Theorem 4.1, \cite{cai_dichotomy_2017}]
 A \csp$_2^c(\cF)$ problem has a polynomial time algorithm if one of the following holds:
 \begin{itemize}
  \item $\cF\subseteq\avg{\cE}$,
  \item $\cF\subseteq\cA$,
  \item $\cF\subseteq T\circ\cA$, or
  \item $\cF\subseteq\cL$.
 \end{itemize}
 Otherwise, it is \sP-hard.
\end{thm}

The preceding results all apply to complex-valued signatures, but the following theorem is restricted to real-valued ones.

\begin{thm}[Theorem 5.1, \cite{cai_dichotomy_2017}]
\label{thm:real-valued_Holant-c}
 Let $\cF$ be a set of real-valued signatures. Then $\Holp[c]{\cF}$ is \sP-hard unless $\cF$ is a tractable family for \Hol$^*$ or \csp$_2^c$.
\end{thm}

\subsection{Complexity results for ternary signatures}

In addition to the above-mentioned Holant dichotomies, there are also some dichotomies specific to symmetric signatures on three-regular graphs.
For signature sets containing a ternary GHZ-type signature, there is furthermore a direct relationship to \csp{}, which allows a more general complexity classification.
When deriving the \Hol$^c$ dichotomy, our general approach will be to attempt to construct a gadget for an entangled ternary signature and then use the following results.

\begin{thm}[Theorem 3.4, \cite{cai_holant_2012}]\label{thm:W-state}
 $\Holant([y_0,y_1,y_2]|[x_0,x_1,x_2,x_3])$ is \sP-hard unless $[x_0,x_1,x_2,x_3]$ and $[y_0,y_1,y_2]$ satisfy one of the following conditions, in which case the problem is in \FP:
 \begin{itemize}
  \item $[x_0,x_1,x_2,x_3]$ is degenerate, or
  \item there is a 2 by 2 matrix $M$ such that:
   \begin{itemize}
    \item $[x_0,x_1,x_2,x_3]=M\circ[1,0,0,1]$ and $(M^T)^{-1}\circ[y_0,y_1,y_2]$ is in $\mathcal{A}\cup\avg{\mathcal{E}}$,
    \item $[x_0,x_1,x_2,x_3]=M\circ[1,1,0,0]$ and $(M^T)^{-1}\circ[y_0,y_1,y_2]$ is of the form $[0,*,*]$,
    \item $[x_0,x_1,x_2,x_3]=M\circ[0,0,1,1]$ and $(M^T)^{-1}\circ[y_0,y_1,y_2]$ is of the form $[*,*,0]$,
   \end{itemize}
 \end{itemize}
 with $*$ denoting an arbitrary complex number.
\end{thm}

The signature $\ket{000}+\ket{111}$ is invariant under holographic transformations of the form $\left(\begin{smallmatrix}1&0\\0&\omega\end{smallmatrix}\right)$, where $\omega^3=1$.
Therefore, a binary signature is considered to be $\omega$-normalised if $y_0=0$, or there does not exist a primitive $(3t)$-th root of unity $\lambda$, where $gcd(t,3)=1$, such that $y_2=\lambda y_0$.
Similarly, a unary signature $[a,b]$ is $\omega$-normalised if $a=0$, or there does not exist a primitive $(3t)$-th root of unity $\lambda$, where $gcd(t,3)=1$, such that $b=\lambda a$.

\begin{thm}[Theorem 4.1, \cite{cai_holant_2012}]
 \label{thm:GHZ-state}
 Let $\mathcal{G}_1,\mathcal{G}_2$ be two sets of signatures and let $[y_0,y_1,y_2]$ be a $\omega$-normalised and non-degenerate signature.
 In the case of $y_0=y_2=0$, further assume that $\mathcal{G}_1$ contains a unary signature $[a,b]$ which is $\omega$-normalised and satisfies $ab\neq 0$.
 Then:
 \begin{equation}
  \Holp{\{[y_0,y_1,y_2]\}\cup\mathcal{G}_1 \mid \{[1,0,0,1]\}\cup\mathcal{G}_2} \equiv_T \csp(\{[y_0,y_1,y_2]\}\cup\mathcal{G}_1\cup\mathcal{G}_2).
 \end{equation}
 More specifically, $\Holp{\{[y_0,y_1,y_2]\}\cup\mathcal{G}_1 \mid \{[1,0,0,1]\}\cup\mathcal{G}_2}$ is \sP-hard unless:
 \begin{itemize}
  \item $\{[y_0,y_1,y_2]\}\cup\mathcal{G}_1\cup\mathcal{G}_2\subseteq\avg{\cE}$, or
  \item $\{[y_0,y_1,y_2]\}\cup\mathcal{G}_1\cup\mathcal{G}_2\subseteq\mathcal{A}$,
 \end{itemize}
 in which cases the problem is in \FP.
\end{thm}

The following lemmas show how to construct symmetric ternary entangled signatures from non-symmetric ones.

\begin{lem}[Lemma 18, \cite{backens_new_2017}] \label{lem:GHZ_symmetrise}
 Let $\ket{\psi}$ be a ternary GHZ-type signature, i.e.\ $\ket{\psi}=(A\otimes B\otimes C)\ket{\GHZ}$ for some invertible 2 by 2 matrices $A,B,C$. Then at least one of the three possible symmetric triangle gadgets constructed from three copies of $\ket{\psi}$ is non-degenerate, unless $\ket{\psi}\in K\circ\cE$ and is furthermore already symmetric.
\end{lem}

\begin{lem}[Lemma 19, \cite{backens_new_2017}] \label{lem:W_symmetrise}
 Let $\ket{\psi}$ be a ternary $W$-type signature, i.e.\ $\ket{\psi}=(A\otimes B\otimes C)\ket{W}$ for some invertible 2 by 2 matrices $A,B,C$. If $\ket{\psi}\in K\circ\cM$ (or $\ket{\psi}\in KX\circ\cM$), assume that we also have a binary entangled signature $\ket{\phi}$ that is not in $K\circ\cM$ (or $KX\circ\cM$, respectively). Then we can construct a symmetric ternary entangled signature.
\end{lem}

\subsection{Results about binary and 4-ary signatures}

Besides the above results about ternary signatures, we will also make use of the following result about realising or interpolating the 4-ary equality signature from a more general 4-ary signature.

\begin{lem}[Lemma 2.38, \cite{cai_holographic_2016}]\label{lem:interpolate_equality4}
 Suppose $\cF$ contains a signature $f$ of arity 4 with:
 \[
  \ket{f} = a\ket{0000}+b\ket{0011}+c\ket{1100}+d\ket{1111},
 \]
 where $M=\left(\begin{smallmatrix}a&b\\c&d\end{smallmatrix}\right)$ has full rank. Then:
 \[
  \textsc{Pl-Holant}(\{=_4\}\cup\cF) \leq_T \textsc{Pl-Holant}(\cF).
 \]
\end{lem}

Here, \textsc{Pl-Holant} refers to the Holant problem for planar graphs.
The lemma can of course also be used in the non-planar setting.

Sometimes we can realise a 4-ary generalised equality, i.e.\ a signature of the form $a\ket{x_1 x_2 x_3 x_4}+b\ket{\bar{x}_1 \bar{x}_2 \bar{x}_3 \bar{x}_4}$ with $a,b\in\CC\setminus\{0\}$ and $x_k\in\{0,1\}$ for $k\in[4]$.
Then the following lemma is helpful.

\begin{lem}[Lemma 5.2, \cite{cai_dichotomy_2017}]\label{lem:generalised_equality4}
 Suppose $\cF$ contains a 4-ary generalised equality, then $\Holp{\cF}\equiv_T\csp_2(\cF)$.
\end{lem}

The next result arises from the study of entanglement in quantum information theory; we translate it into the terminology and notation used throughout this paper.

\begin{thm}[\cite{popescu_generic_1992,gachechiladze_completing_2017,backens_new_2017}]
\label{thm:popescu_rohrlich}
 Let $\ket{\psi}$ be an $n$-ary genuinely entangled signature with $n\geq 2$.
 For any choice of inputs $j,k\in[n]$ there exist unary signatures $\ket{\phi_l}\in\{\ket{0},\ket{1},\ket{\pm}\}$ with $l\in [n]\setminus\{j,k\}$ such that $\left( \bigotimes_{l} \bra{\phi_l}_l \right) \ket{\psi}$ is a binary entangled signature.
\end{thm}

In other words, given a genuinely entangled $n$-ary signature and the unary signatures $\ket{0},\ket{1},\ket{+}$ and $\ket{-}$, we can realise a binary entangled signature.

\section{The dichotomy}
\label{s:dichotomy}

In \cite{cai_dichotomy_2017}, Cai \etal{} derive a dichotomy for real-valued \Hol$^c$.
Much of their proof -- in particular the dichotomy for \csp$_2^c$, which is crucial to the \Hol$^c$ dichotomy -- applies to complex-valued signatures.
They only switch to considering real-valued signatures towards the end of the paper.
There appear to be two barriers to extending their dichotomy to complex values: firstly, some of the hardness results for not necessarily symmetric ternary entangled signatures in \cite{cai_dichotomy_2017} use techniques that only apply to real values.
Secondly, some cases of the proof of the main theorem rely on being able to interpolate all unary signatures, again using techniques that have only been shown to work for real-valued signatures.

We avoid these barriers by adapting techniques from the dichotomy for \Hol$^+$ \cite{backens_new_2017}: in particular, techniques for realising symmetric ternary entangled signatures from not-necessarily symmetric ones (cf.\ Lemmas \ref{lem:GHZ_symmetrise} and \ref{lem:W_symmetrise}).
Because of this, we never need to interpolate arbitrary unary signatures.
In one subcase, we may require several unary signatures other than $\ket{0}$ and $\ket{1}$, but we give a new construction for realising sufficiently many such signatures by gadgets.
We also point out that one of the interpolations used for real-valued signatures in the proof of Theorem \ref{thm:real-valued_Holant-c} can straightforwardly be extended to complex-valued signatures by Lemma \ref{lem:interpolate_equality4}.

The general proof strategy in the dichotomy theorem is to attempt to realise an entangled ternary signature and then a symmetric entangled ternary signature.
In some cases, the symmetric constructions and subsequent hardness proofs require a little extra work in the \Hol$^c$ setting as compared to the \Hol$^+$ setting; we deal with those issues in Section \ref{s:ternary}.

If we cannot construct a ternary entangled signature, either $\cF\subseteq\avg{\cT}$, in which case the problem is tractable by Theorem \ref{thm:Holant-star}, or we can construct a 4-ary entangled signature of a specific form.
Again, this is analogous to the real-valued case in \cite{cai_dichotomy_2017}.
Given this 4-ary signature, we realise or interpolate a 4-ary equality signature by Lemma \ref{lem:interpolate_equality4}, which reduces the problem to $\csp_2^c$ by Lemma \ref{lem:generalised_equality4}.

The main theorem and its proof are given in Section \ref{s:main_theorem}.

\subsection{Hardness proofs involving a ternary entangled signature}
\label{s:ternary}

First, we prove several lemmas that give a complexity classification for \Hol$^c$ problems in the presence of a ternary entangled signature.
The following results supersede Lemmas 5.1, 5.3, and 5.5--5.7 of \cite{cai_dichotomy_2017}.
Whereas the last three of those only apply to real-valued signatures, the new results work for complex values, too.

\begin{lem}\label{lem:arity3_hardness}
 Let $\ket{\psi}\in\cF$ be an entangled ternary signature.
 Then $\Holp[c]{\cF}$ is \sP-hard unless:
 \begin{itemize}
  \item $\Holp[*]{\cF}$ is tractable, or
  \item $\cF\subseteq S\circ\cA$ for some $S\in\cS$, as defined in \eqref{eq:cS_definition}.
 \end{itemize}
 In both of those cases, the problem $\Holp[c]{\cF}$ is tractable too.
\end{lem}
\begin{proof}
 First, suppose $\ket{\psi}$ is symmetric.
 
 If $\ket{\psi}$ is of GHZ type, insert an extra vertex in the middle of each edge and assign it the signature $\ket{00}+\ket{11}$.
 This leaves the Holant invariant, i.e.:
   \begin{align*}
    \Holp[c]{\cF} &\equiv_T \Holp{\cF\cup\{\ket{0},\ket{1}\}} \\
    &\equiv_T \Holp{\cF\cup\{\ket{0},\ket{1}\}\mid\{\ket{00}+\ket{11}\}}
   \end{align*}
   Furthermore, we have:
   \[
    \Holp{\cF\cup\{\ket{0},\ket{1}\}\mid\{\ket{00}+\ket{11}\}} \equiv_T \Holp{\cF\cup\{\ket{0},\ket{1}\}\mid\{\ket{00}+\ket{11},\ket{0},\ket{1}\}}.
   \]
   The $\leq_T$ direction is immediate; for the other direction, note that any occurrence of $\ket{0}$ or $\ket{1}$ on the RHS can be replaced by a gadget consisting of a LHS copy of the unary signature connected to $\ket{00}+\ket{11}$.
   
   Let $M$ be an invertible 2 by 2 complex matrix such that $M\t{3}\ket{\psi}=\ket{\GHZ}$ and $((M^{-1})^T)\t{2}(\ket{00}+\ket{11})$ is $\omega$-normalised.
   Then, if $((M^{-1})^T)\t{2}(\ket{00}+\ket{11})$ is not of the form $c(\ket{01}+\ket{10})$ for some $c\in\CC$, by Theorem \ref{thm:GHZ-state}:
   \begin{equation}\label{eq:Holant-c_csp}
    \Holp[c]{\cF} \equiv_T \csp\left( M\circ\left( \cF \cup \left\{ \ket{0}, \ket{1} \right\} \right) \cup \left(M^{-1}\right)^T \circ\{ \ket{00}+\ket{11}, \ket{0}, \ket{1} \} \right)
   \end{equation}
   Now, $((M^{-1})^T)\t{2}(\ket{00}+\ket{11})=c(\ket{01}+\ket{10})$ implies that $M^{-1}=KD$ or $M^{-1}=KXD$ for some invertible diagonal matrix $D$, as shown in \cite{backens_new_2017}.
   In either case, all components of $M^{-1}$ are non-zero, so $(M^{-1})^T\ket{0}=a\ket{0}+b\ket{1}$ for some $a,b,\in\CC\setminus\{0\}$.
   It is then possible to choose $M$ so that $(M^{-1})^T\ket{0}$ is $\omega$-normalised, satisfying the condition of Theorem \ref{thm:GHZ-state}.
   Hence, again, \eqref{eq:Holant-c_csp} holds.
   In either case, $M$ is determined up to non-zero scalar factor and post-multiplication by $X$, which do not affect any of the subsequent arguments.
   
   Thus, $\Holp[c]{\cF}$ is \sP-hard unless:
   \[
    M\circ\left( \cF \cup \left\{ \ket{0}, \ket{1} \right\} \right) \cup \left(M^{-1}\right)^T \circ\{\ket{00}+\ket{11}, \ket{0}, \ket{1}\}
   \]
   is a subset of either $\avg{\cE}$ or $\cA$.
   All entangled binary signatures in $\avg{\cE}$ are of the form $\alpha\ket{00}+\beta\ket{11}$ or $\alpha\ket{01}+\beta\ket{10}$ for some $\alpha,\beta\in\CC\setminus\{0\}$.
   Hence the signatures can be in $\avg{\cE}$ only if $M=DO$ for some diagonal matrix $D$ and orthogonal matrix $O$, or $M=DK^T$ or $M=DXK^T$.
   In those cases, $\cF\subseteq\avg{O\circ\cE}$ or $\cF\subseteq\avg{K\circ\cE}$, so we have tractability by the \Hol$^*$ dichotomy.
  
   For the signatures to be in $\cA$, we require in particular:
   \begin{equation}\label{eq:condition_cA}
    ((M^{-1})^T)\t{2}(\ket{00}+\ket{11}),\;(M^{-1})^T\ket{0},\; (M^{-1})^T\ket{1}\in\cA.
   \end{equation}
   That is just the definition of $M^{-1}\in\cS$, cf.\ \eqref{eq:cS_definition}.
   We also require $M\ket{0}, M\ket{1}\in\cA$.
   Now if $M=\left(\begin{smallmatrix}a&b\\c&d\end{smallmatrix}\right)$ with $a,b,c,d\in\CC$ and $ad-bc\neq 0$, then $M\ket{0}=a\ket{0}+c\ket{1}$ and $(M^{-1})^T\ket{1}=\frac{1}{ad-bc}(-c\ket{0}+a\ket{1})$, so $M\ket{0}$ and $(M^{-1})^T\ket{1}$ are orthogonal to each other (under the inner product that does not involve complex conjugation).
   The same holds for $M\ket{1}$ and $(M^{-1})^T\ket{0}$.
   It is straightforward to see that if $\ket{\phi}\in\cA$ is unary and $(\ket{\phi^\perp})^T\ket{\phi}=0$, then $\ket{\phi^\perp}$ is also affine.
   So $M^{-1}\in\cS$ already implies $M\ket{0}, M\ket{1}\in\cA$.
   The remaining condition is $M\circ\cF\subseteq\cA$; since $M$ is invertible, this is equivalent to $\cF\subseteq M^{-1}\circ\cA$.
   
   To conclude, in the GHZ case the problem is tractable if $\cF\subseteq\avg{O\circ\cE}$ for some orthogonal 2 by 2 matrix $O$, if $\cF\subseteq\avg{K\circ\cE}$, or if there exists $S\in\cS$ such that $\cF\subseteq S\circ\cA$.
   In all other cases, the problem is \sP-hard by reduction from \csp.

 If $\ket{\psi}$ is of $W$ type, then:
   \begin{itemize}
    \item If $\ket{\psi}\notin K\circ\cM\cup KX\circ\cM$, $\Holp{\ket{\psi}}$ is \sP-hard by Theorem \ref{thm:W-state}.
    \item If $\cF\subseteq K\circ\cM$ or $\cF\subseteq KX\circ\cM$, the problem is tractable by the \Hol$^*$ dichotomy.
    \item If $\ket{\psi}\in K\circ\cM$ but $\cF\not\subseteq K\circ\cM$, the problem is \sP-hard by Lemma \ref{lem:case_KM}, and analogously with $KX$ instead of $K$.
   \end{itemize}
 
 Now assume $\ket{\psi}$ is not symmetric.
 If $\ket{\psi}\notin K\circ\cM\cup KX\circ\cM$, we can construct a symmetric ternary signature by Lemmas \ref{lem:GHZ_symmetrise} and \ref{lem:W_symmetrise} and then proceed as above.
 If $\cF\subseteq K\circ\cM$ or $\cF\subseteq KX\circ\cM$, the problem is tractable by the \Hol$^*$ dichotomy.
 
 Finally, if $\ket{\psi}\in K\circ\cM$ but $\cF\not\subseteq K\circ\cM$, or $\ket{\psi}\in KX\circ\cM$ but $\cF\not\subseteq KX\circ\cM$, use Lemma \ref{lem:case_KM}.
 
 This covers all cases.
\end{proof}

\begin{lem}\label{lem:case_KM}
 Let $\ket{\psi}\in\cF\cap K\circ\cM$ be a ternary entangled signature, and assume $\cF\not\subseteq K\circ\cM$.
 Then $\Holp[c]{\cF}$ is \sP-hard.
 The same holds if $\ket{\psi}\in\cF\cap KX\circ\cM$ and $\cF\not\subseteq KX\circ\cM$.
\end{lem}
\begin{proof}
 We consider the case $\ket{\psi}\in\cF\cap K\circ\cM$ and $\cF\not\subseteq K\circ\cM$, the proof for the second case is analogous.
 
 As $\cF\not\subseteq K\circ\cM$, we can find $\ket{\varphi}\in\cF\setminus K\circ\cM$.
 Then $\ket{\varphi}$ has arity at least 2, as all unary signatures are in $K\circ\cM$.
 Without loss of generality, assume $\ket{\varphi}$ is genuinely entangled.
 
 If $\ket{\varphi}$ has arity 2, we can realise a symmetric ternary entangled signature by Lemma \ref{lem:W_symmetrise}.
 If that symmetric ternary signature is of GHZ type, we can apply Theorem \ref{thm:GHZ-state}.
 If it is of $W$ type and not in $K\circ\cM\cup KX\circ\cM$, the problem is \sP-hard by Theorem \ref{thm:W-state}.
 
 Finally, if the symmetric ternary signature is in $K\circ\cM$, then we can use it and $\ket{\varphi}$ to realise a symmetric binary signature that is not in 
 $K\circ\cM$ by an argument analogous to that in Lemma 21 of \cite{backens_new_2017}.
 While that Lemma assumes availability of four unary signatures, there is only one value for which the construction fails, so it suffices to have two unary signature available.
 With a symmetric binary signature that is not in $K\circ\cM$, hardness follows by Theorem \ref{thm:W-state}.
 
 Now assume $\ket{\varphi}$ is an $n$-ary entangled signature with $n>2$.
 Note that we can write the ternary signature $\ket{\psi}$ as $K\t{3}(a\ket{000}+b\ket{001}+c\ket{010}+d\ket{100})$, where $bcd\neq 0$.
 With a self loop on a vertex assigned signature $\ket{\psi}$, we can therefore construct unary signatures $2(b+c)(\ket{0}+i\ket{1}$, $2(b+d)(\ket{0}+i\ket{1})$ and $2(c+d)(\ket{0}+i\ket{1})$.
 As $bcd\neq 0$, at least one of those gadgets is non-zero.
 Thus we can realise $\ket{0}+i\ket{1}$.
 
 The remainder of the argument will be more straightforward after a holographic transformation. 
 We have:
 \[
  \Holp[c]{\cF} = \Holp{\cF\cup\{\ket{0},\ket{1}\}} \equiv_T \Holp{\cF\mid\{\ket{0},\ket{1},\ket{00}+\ket{11}\}}
 \]
 as subgraphs that contain only the signatures $\ket{0},\ket{1},\ket{00}+\ket{11}$ must contribute a factor 0 or 1 to the Holant, and it is straightforward to determine that factor.
 Therefore, by a holographic transformation:
 \begin{align*}
  \Holp[c]{\cF} &\equiv_T \Holp{ K^{-1}\circ\cF \mid K^T\circ\{\ket{0},\ket{1},\ket{00}+\ket{11}\} } \\
  &\equiv_T \Holp{ K^{-1}\circ\cF \mid \{\ket{+},\ket{-},\ket{01}+\ket{10}\} }.
 \end{align*}
 $As \ket{\varphi},\ket{\psi}\in\cF$, after the transformation, we have:
 \[
  \ket{\psi'}=(K^{-1})\t{3}\ket{\psi}=a\ket{000}+b\ket{001}+c\ket{010}+d\ket{100}
 \]
 and $\ket{\varphi'}=K^{-1}\circ\ket{\varphi}$ on the LHS.
 On the RHS, in addition to the already listed signatures, we have $K^T(\ket{0}+i\ket{1}) \doteq \ket{1}$ from the self-loop gadget described above.
 
 By plugging $\ket{+}$ or $\ket{-}$ into $\ket{\psi'}$, we can realise a LHS gadget with signature:
 \[
  (a\pm b)\ket{00} + c\ket{01} + d\ket{10}.
 \]
 There is a choice of sign making the coefficient of $\ket{00}$ non-zero.
 Then two such gadgets can be combined into a symmetric one, which (up to scalar factor) has signature $z\ket{00}+\ket{01}+\ket{10}$, where $z=\frac{2}{c}(a\pm b)\neq 0$.
 A chain of $n$ of these gadgets connected to $\ket{\pm}$ at one end gives a RHS gadget with signature:
 \[
  \ket{0}+ (nz\pm 1)\ket{1}.
 \]
 Whatever the value of $z$, we can realise polynomially many different unary signatures on the RHS.
 These can be used to realise a binary entangled LHS gadget by the following argument, which is similar to Lemma 20 in \cite{backens_new_2017}.
 Nevertheless, the differences are significant enough to give the entire proof here.
 
 From Theorem \ref{thm:popescu_rohrlich}, we know that there exist $\ket{\phi_k}\in\{\ket{0},\ket{1},\ket{\pm}\}$ for $k\in\{3,4,\ldots,n\}$ such that $\bra{\phi_3}_3\ldots\bra{\phi_n}_n \ket{\varphi'}$ is entangled.
 The entanglement condition for binary signatures is:
 \begin{multline}
  \bra{0}_1\bra{0}_2\bra{\phi_3}_3\ldots\bra{\phi_n}_n \ket{\varphi'} \bra{1}_1\bra{1}_2\bra{\phi_3}_3\ldots\bra{\phi_n}_n \ket{\varphi'} \\
  - \bra{0}_1\bra{1}_2\bra{\phi_3}_3\ldots\bra{\phi_n}_n \ket{\varphi'} \bra{1}_1\bra{0}_2\bra{\phi_3}_3\ldots\bra{\phi_n}_n \ket{\varphi'} \neq 0.
  \label{eq:entangled}
 \end{multline}
 Furthermore, as $\ket{\varphi'}\notin\cM$, there exists a bit string $y$ of Hamming weight at least 2 such that $\braket{y}{\varphi'}\neq 0$.
 Without loss of generality, assume $y_1=y_2=1$.
 Then:
 \begin{equation}\label{eq:not_cM}
  \bra{1}_1\bra{1}_2\bra{y_3}_3\ldots\bra{y_n}_n \ket{\varphi'} \neq 0.
 \end{equation}
 for some $y_3,\ldots,y_k\in\{0,1\}$.
 
 We show how to realise a binary gadget whose signature is not in $K\circ\cM$.
 Consider the inputs of $\ket{\varphi'}$ one by one, starting with the third.
 For the $k$-th input:
 \begin{itemize}
  \item If $\ket{\phi_k}=\ket{y_k}=\ket{1}$, leave it and move on to the $(k+1)$-th input.
  \item Otherwise, replace both $\ket{\phi_k}$ and $\ket{y_k}$ with $\ket{0}+\alpha_k\ket{1}$, where $\alpha_k\in\CC\setminus\{0\}$ is as yet undetermined.
  Then the LHS of \eqref{eq:not_cM} is a linear polynomial in $\alpha_k$, which does not vanish identically, and the LHS of \eqref{eq:entangled} is a quadratic polynomial in $\alpha_k$, which also does not vanish identically.
  Thus, there are at most three values of $\alpha_k$ for which one or both of the polynomials are zero.
  We can therefore realise a signature $\ket{0}+\alpha_k\ket{1}$ such that both polynomials are non-zero.
  Replace $\ket{\phi_k}$ and $\ket{y_k}$ by this new signature.
  Then move on to the $(k+1)$-th input.
 \end{itemize}
 Once all (but the first two) inputs have been considered in this way, we have a recipe for a gadget construction whose signature is binary, entangled, and not in $K\circ\cM$.
 Therefore we can proceed as in the case where $\ket{\varphi}$ is binary.
\end{proof}

\subsection{Main theorem}
\label{s:main_theorem}

We now have all the components required to prove the main dichotomy for \Hol$^c$.
The theorem generalises Theorem 5.1 of \cite{cai_dichotomy_2017}, which applies only to real-valued signatures.
Our proof follows the original one fairly closely.

\begin{thm}
 Let $\cF$ be a set of complex-valued signatures.
 Then $\Holp[c]{\cF}$ is \sP-hard unless:
 \begin{itemize}
  \item $\cF$ is a tractable family for \Hol$^*$,
  \item there exists $S\in\cS$ such that $\cF\subseteq S\circ\cA$, or
  \item $\cF\subseteq\cL$.
 \end{itemize}
 In all of the exceptional cases, $\Holp[c]{\cF}$ is tractable.
\end{thm}
\begin{proof}
 If $\cF$ is one of the tractable families for \Hol$^*$ or $\cF\subseteq S\circ\cA$ for some $S\in\cS$, or $\cF\subseteq\cL$, tractability of $\Holp[c]{\cF}$ follows using the same algorithms as employed in the dichotomy proofs for \Hol$^*$, \csp{} (possibly after a holographic transformation), or \csp$_2^c$.
 So assume otherwise.
 In particular, this implies that $\cF\not\subseteq\avg{\cT}$, i.e.\ $\cF$ has multipartite entanglement.
 
 Without loss of generality, we may focus on genuinely entangled signatures. 
 So assume that there is some genuinely entangled signature $\ket{\psi}\in\cF$ of arity $n\geq 3$.
 If the signature has arity 3, we are done by Lemma \ref{lem:arity3_hardness}.
 Hence assume $n\geq 4$.
 As $\ket{\psi}$ is genuinely entangled, there exist two distinct $n$-bit strings $x$ and $y$ such that $\braket{x}{\psi}\braket{y}{\psi}\neq 0$.
 
 As in \cite{cai_dichotomy_2017}, let:
 \[
  D_0 = \min \left\{ d(x,y) \mid x\neq y, \braket{x}{\psi}\neq 0, \braket{y}{\psi}\neq 0 \right\},
 \]
 where $d(\cdot,\cdot)$ is the Hamming distance, and distinguish cases according to the values of $D_0$.
 
 \textbf{Case $D_0\geq 4$ and $D_0$ is even}: Pick a pair of bit strings $x,y$ with minimal Hamming distance. 
 Pin all inputs where the two bit string agree (without loss of generality, we always assume bit strings agree on the last $n-D_0$ bits).
 This realises a signature of the form  $a(\bigotimes_{k=1}^{D_0} \ket{x_k}) +b(\bigotimes_{k=1}^{D_0} \ket{\bar{x}_k})$, where $ab\neq 0$ and $x_k\in\{0,1\}$ for $k\in[D_0]$.
 Via self-loops, the arity of this signature can be reduced in steps of 2 to realise a 4-ary generalised equality signature, i.e.\ a signature of the form $a\ket{x_1 x_2 x_3 x_4}+b\ket{\bar{x}_1 \bar{x}_2 \bar{x}_3 \bar{x}_4}$.
 Then \sP-hardness follows by Lemma \ref{lem:generalised_equality4}.
 
 \textbf{Case $D_0\geq 3$ and $D_0$ is odd}: Pin analogously to the previous case to realise a signature $a\ket{x_1x_2x_3}+b\ket{\bar{x}_1\bar{x}_2\bar{x}_3}$ with $ab\neq 0$ and $x_1,x_2,x_3\in\{0,1\}$.
 Then \sP-hardness follows by Lemma \ref{lem:arity3_hardness}.
 
 \textbf{Case $D_0=2$}: We can realise a signature $a\ket{00}+b\ket{11}$ by pinning.
 Following the proof in \cite{cai_dichotomy_2017}, let $A_1$ be the set of bit strings $x = x_3x_4\ldots x_n$ for which $\ket{\phi_x}:=(\bra{x_3}_3\ldots\bra{x_n}_n)\ket{\psi}$ is a non-zero scaling of $a\ket{00}+b\ket{11}$.
 Let $B_1$ be the set of bit strings $y=y_3\ldots y_n$ for which $\ket{\phi_y}$ is not a scaling of $a\ket{00}+b\ket{11}$.
 This excludes bit strings for which $\ket{\phi_y}=0$.
 Both $A_1$ and $B_1$ must be non-empty as $\ket{\psi}$ is entangled.
 Furthermore, $A_1\cap B_1=\emptyset$.
 Thus we can define:
 \[
  D_1 = \min\{ d(x,y) \mid x \in A_1, y\in B_1 \}.
 \]
 Note that the assumption $D_0=2$ implies that either $\braket{01}{\phi_y}=\braket{10}{\phi_y}=0$ or $\braket{00}{\phi_y}=\braket{11}{\phi_y}=0$ for all $y\in B_1$.
 We now distinguish cases according to the values of $D_1$.
 \begin{itemize}
  \item If $D_1\geq 3$, pick a pair $x,y$ with minimal Hamming distance and pin wherever they are equal, as in the cases where $D_0\geq 3$.
   This realises a signature:
   \[
    (a\ket{00}+b\ket{11}) \left(\bigotimes_{k=1}^{D_1} \ket{x_k}\right) + (c\ket{00}+d\ket{11}) \left(\bigotimes_{k=1}^{D_1} \ket{\bar{x}_k}\right)
   \]
   or
   \[
    (a\ket{00}+b\ket{11}) \left(\bigotimes_{k=1}^{D_1} \ket{x_k}\right) + (c\ket{01}+d\ket{10}) \left(\bigotimes_{k=1}^{D_1} \ket{\bar{x}_k}\right),
   \]
   where $x_k\in\{0,1\}$ for $k\in[D_1]$, $c,d\in\CC$ are not both zero, and, in the first case, $ad-bc\neq 0$.
   \begin{itemize}
    \item In the first case, suppose $c\neq 0$.
     Then we can pin the first two inputs to 00 to get a signature $a(\bigotimes_{k=1}^{D_1} \ket{x_k})+c(\bigotimes_{k=1}^{D_1} \ket{\bar{x}_k})$, at which point we proceed as in the cases $D_0\geq 4 $ or $D_0\geq 3$.
     If $c=0$ then $d\neq 0$ and we can pin to 11 instead for an analogous argument.
    \item In the second case, suppose $c\neq 0$.
     Then we can pin the first input to 0 to realise:
     \[
      a\ket{0}\left(\bigotimes_{k=1}^{D_1} \ket{x_k}\right)+c\ket{1}\left(\bigotimes_{k=1}^{D_1} \ket{\bar{x}_k}\right)
     \]
     at which point we again proceed as in the cases $D_0\geq 4 $ or $D_0\geq 3$.
     If $c=0$ then $d\neq 0$ and we can pin to 1 instead for an analogous argument.
   \end{itemize}
  \item If $D_1=2$ and the signature after projecting is $a\ket{0000}+b\ket{1100}+c\ket{0111}+d\ket{1011}$, we can realise a ternary signature as in the second subcase above, and then apply Lemma \ref{lem:arity3_hardness}.
   If the signature after pinning is $a\ket{0000}+b\ket{1100}+c\ket{0011}+d\ket{1111}$ with $ab\neq 0$ and $ad-bc\neq 0$, apply Lemma \ref{lem:interpolate_equality4}.
   Here, the original proof in \cite{cai_dichotomy_2017} used a different technique requiring real values.
  \item If $D_1=1$, we can realise an entangled ternary signature and then apply Lemma \ref{lem:arity3_hardness}.
 \end{itemize}
   
 \textbf{Case $D_0=1$}: We can realise $a\ket{0}+b\ket{1}$ for some $ab\neq 0$.
 Let $A_2$ be the set of bit strings $x = x_2 x_3\ldots x_n$ for which $\ket{\varphi_x}:=(\bra{x_2}_2\ldots\bra{x_n}_n)\ket{\psi}$ is a non-zero scaling of $a\ket{0}+b\ket{1}$.
 Let $B_2$ be the set of bit strings $y$ for which $\ket{\varphi_y}$ is not a scaling of $a\ket{0}+b\ket{1}$.
 Then let:
 \[
  D_2 = \min\{ d(x,y) \mid x \in A_2, y\in B_2 \}.
 \]
 \begin{itemize}
  \item If $D_2\geq 3$, we can pin to realise a signature:
    \[
     (a\ket{0}+b\ket{1}) \left(\bigotimes_{k=1}^{D_2} \ket{x_k}\right) + (c\ket{0}+d\ket{1}) \left(\bigotimes_{k=1}^{D_2} \ket{\bar{x}_k}\right)
    \]
    where $c,d\in\CC$, $ad-bc\neq 0$, and $x_k\in\{0,1\}$ for $k\in[D_2]$.
    If $c\neq 0$, pin the first input to 0 and then proceed as before.
    If $c=0$ then $d\neq 0$, so we can pin to 1 instead.
  \item If $D_2=2$, we get an entangled ternary signature so we are done by Lemma \ref{lem:arity3_hardness}.
   This is another change compared to the proof in \cite{cai_dichotomy_2017}, where hardness was only shown for a real-valued signature of the given form.
  \item If $D_2=1$, we can realise an entangled binary signature $a\ket{00}+b\ket{01}+c\ket{10}+d\ket{11}$ with $ad-bc\neq 0$.
  Unlike in \cite{cai_dichotomy_2017}, we do not attempt to use this binary signature for interpolation.
  Instead we immediately proceed to defining $A_3$, $B_3$, and $D_3$ analogous to before.
   \begin{itemize}
    \item If $D_3\geq 3$, we can realise a signature:
     \begin{multline*}
      (a\ket{00}+b\ket{01}+c\ket{10}+d\ket{11}) \left(\bigotimes_{k=1}^{D_3} \ket{x_k}\right) \\
      + (a'\ket{00}+b'\ket{01}+c'\ket{10}+d'\ket{11}) \left(\bigotimes_{k=1}^{D_3} \ket{\bar{x}_k}\right)
     \end{multline*}
     where $(a',b',c',d')$ and $(a,b,c,d)$ are linearly independent.
     If $aa'=bb'=cc'=dd'=0$, then it must be the case that $c=0$ and $c'd\neq 0$, or $d=0$ and $cd'\neq 0$.
     This is because $a',b',c',d'$ cannot all be zero simultaneously, and $a,b$ are non-zero by assumption.
     In the former case, pin the first input to 1 to get:
     \[
      d\ket{1}\left(\bigotimes_{k=1}^{D_3} \ket{x_k}\right)+c'\ket{0}\left(\bigotimes_{k=1}^{D_3} \ket{\bar{x}_k}\right),
     \]
     then proceed as before.
     In the latter case, the same approach works, although the resulting signature is different.
       
     Otherwise, there exists a pair of primed and unprimed coefficients of the same label that are both non-zero.
     If these are $a$ and $a'$, pin the first two inputs to 00 to get a generalised equality.
     If the non-zero pair are $b$ and $b'$, pin to 01, and so on.
    \item If $D_3=2$, as shown in the original theorem, we can realise the following signature:
     \[
      (a\ket{00}+b\ket{01}+c\ket{10}+d\ket{11})\ket{x_1x_2} + (a'\ket{00}+b'\ket{01}+c'\ket{10}+d'\ket{11})\ket{\bar{x}_1\bar{x}_2},
     \]
     where $ab\neq 0$, $ad-bc\neq 0$, $x_1,x_2\in\{0,1\}$, and $(a',b',c',d')$ is linearly independent from $(a,b,c,d)$.
     We can realise a genuinely entangled ternary signature by connecting $a\ket{0}+b\ket{1}$ to the last input, at which point we can apply Lemma \ref{lem:arity3_hardness}.
     This is a change compared to the original proof in \cite{cai_dichotomy_2017}, where the hardness lemma only applied to real values and the construction did not employ the signature $a\ket{0}+b\ket{1}$.
    \item If $D_3=1$, we get an entangled ternary signature so we are done by Lemma \ref{lem:arity3_hardness}.
     This is another change from \cite{cai_dichotomy_2017}, where multiple cases were distinguished and the hardness lemmas only applied to real-valued signatures.
   \end{itemize}
 \end{itemize}

 We have covered all cases, hence the proof is complete.
\end{proof}

\section{Conclusions}
\label{s:conclusions}

Building on the existing dichotomies for real-valued \Hol$^c$ and for complex-valued \Hol$^+$, we have derived a dichotomy for complex-valued \Hol$^c$.
The tractable cases are the complex generalisations of the tractable cases of the real-valued \Hol$^c$ dichotomy.
The question of a dichotomy for complex-valued, not necessarily symmetric \Hol$^c$ problems had been open since the definition of \Hol$^c$ in 2009.

Several steps in the dichotomy proof use knowledge from quantum information theory, particularly about entanglement.
We expect this approach of bringing together Holant problems and quantum information theory to yield further insights into both areas of research in the future.
The ultimate goals include a dichotomy for general Holant problems on the one hand, building up on existing results for symmetric functions \cite{cai_complete_2013} and non-negative real-valued, not necessarily symmetric functions \cite{lin_complexity_2016}.
On the other hand, we are hoping to gain more understanding of the complexity of classically simulating quantum circuits.

In fact, \Hol$^c$ is a natural setting for the latter as many circuit-based quantum computation schemes assume the availability of $\ket{0}$ and $\ket{1}$, called the computational basis states.
On the other hand, quantum computation usually distinguishes between three phases: preparation of input states (often in the computational basis), unitary transformations, and measurement (again, often in the computational basis).
This distinction is not natural or easy to impose in the Holant framework.
The tractable sets related to $\avg{\cE}$ and $\avg{K\circ\cM}$ do not contain any interesting unitary operations, and $\avg{\cT}$ does not contain any multipartite entanglement.
Therefore the only non-trivial tractable classes of quantum computations arising from the \Hol$^c$ dichotomy are those related to affine signatures.
That family of signatures is known in quantum theory as stabilizer quantum mechanics, and it has been known to be efficiently simulable on classical computers for nearly two decades \cite{gottesman_heisenberg_1998}.

Nevertheless, if it is possible to translate the distinction between state preparation, unitary operations, and measurements into the Holant framework -- e.g.\ by considering directed graphs -- it may still be possible to learn more about the complexity of classically simulating other interesting classes of quantum computations.
It may also be useful to look at \Hol$^c$ on planar graphs, or to consider other restricted classes of graphs.

\section*{Acknowledgements}

I would like to thank Pinyan Lu for pointing out a flaw in the original statement of the main theorem.
Many thanks also to Ashley Montanaro and William Whistler for helpful comments on earlier versions of this paper.
I acknowledge funding from EPSRC via grant EP/L021005/1.

\bibliographystyle{eptcs}
\bibliography{refs}

\end{document}